\documentclass[journal]{IEEEtran}
\usepackage{bbm}
\usepackage{caption}
\usepackage{subfigure}
\usepackage{graphicx}
\usepackage{latexsym}
\usepackage{diagbox}
\usepackage{changepage}
\usepackage[fleqn]{amsmath}
\usepackage{amsmath}
\usepackage{amsfonts}
\usepackage{indentfirst}
\usepackage{CJK}
\usepackage{indentfirst}
\usepackage[varg]{txfonts}
\usepackage{stfloats}
\usepackage{multirow}%
\usepackage{booktabs}
\usepackage{color,soul}
\usepackage{epstopdf}
\usepackage{float}
\usepackage{bm}
\usepackage{cite}
\usepackage{makecell}
\usepackage{array}
\usepackage{bm}
\usepackage{mathtools}
\usepackage{geometry}
\usepackage[linesnumbered,ruled,commentsnumbered,longend]{algorithm2e}
\makeatletter

\newcommand{\Rmnum}[1]{\expandafter\@slowromancap\romannumeral #1@}
\renewcommand{\figurename}{Fig.}
\makeatother
\makeatletter

\newtheorem{lemma}{Lemma}
\newtheorem{proposition}{Proposition}

\newenvironment{proof}[1][Proof]{\begin{trivlist}
		\item[\hskip \labelsep {\itshape #1}]}{\end{trivlist}}

\newcommand{\qed}{\nobreak \ifvmode \relax \else
	\ifdim\lastskip<1.5em \hskip-\lastskip
	\hskip1.5em plus0em minus0.5em \fi \nobreak
	\vrule height0.75em width0.5em depth0.25em\fi}

\ifCLASSINFOpdf
\else
\fi
\begin{document}

\title{Edge Cache-assisted Secure Low-Latency Millimeter Wave Transmission}
\author{Wanming Hao,~\IEEEmembership{Member,~IEEE,} Ming Zeng, Gangcan Sun, and Pei Xiao,~\IEEEmembership{Senior Member,~IEEE}
		\thanks{W. Hao is with the School of Information Engineering, and the Henan Institute of Advanced Technology, Zhengzhou University, Zhengzhou 450001, China, and  also with the National Center for International Joint Research of Electronic Materials and Systems, Zhengzhou 450001, China, and also with  the 5G Innovation Center, Institute of Communication Systems, University of Surrey, Guildford GU2 7XH, U.K.  (Email: iewmhao@zzu.edu.cn).}
		\thanks{M. Zeng is with the Faculty of Science and Engineering, Laval University, Quebec, G1V0A6, Canada, and  also with the Faculty of Engineering and Applied Science, Memorial University, St. Johns, NL A1B 3X9, Canada (E-mail: mzeng@mun.ca).}
	\thanks{G. Sun is with the School of Information Engineering, and Institute of Industrial Technology, Zhengzhou University, Zhengzhou 450001, China  (E-mail: iegcsun@zzu.edu.cn).}
	\thanks{P. Xiao is with the 5G Innovation Center, Institute of Communication Systems, University of Surrey, Guildford GU2 7XH, U.K. (Email: p.xiao@surrey.ac.uk).}
	
}

%



\maketitle
\begin{abstract}
In this paper, we consider an edge cache-assisted millimeter wave cloud radio access network (C-RAN). Each remote radio head (RRH) in the C-RAN has a local cache, which can pre-fetch and store the files requested by the actuators. Multiple RRHs form a cluster to cooperatively serve the actuators, which acquire their required files either from the local caches or from the central processor via multicast fronthaul links. For such a scenario, we formulate a beamforming design problem to minimize the secure transmission delay under transmit power constraint of each RRH. Due to the difficulty of directly solving the formulated problem, we divide it into two independent ones: {\textit{i)}} minimizing the fronthaul transmission delay by jointly optimizing the transmit and receive beamforming;  {\textit{ii)}} minimizing the maximum access transmission delay by jointly designing cooperative beamforming among RRHs. An alternatively iterative algorithm is proposed to solve the first optimization problem. For the latter, we first design the analog beamforming based on the channel state information of the actuators. Then, with the aid of successive convex approximation and $S$-procedure techniques, a semidefinite program (SDP) is formulated, and  an iterative algorithm is proposed  through SDP relaxation. Finally, simulation results are provided to verify the performance of the proposed schemes. 
\end{abstract}


%
\IEEEpeerreviewmaketitle

\section{Introduction}
With rapidly growing service in ultra-high-definition video (UHDV), autonomous driving, and connected vehicles, internet of things (IoTs), etc., how to realize high-rate and low-latency transmissions  is of practical significance for future wireless communication networks~\cite{Peng_IoT_2019}. To this end, cloud radio  access network (C-RAN) has emerged as a promising enabling technology~\cite{Wang_TWC_2018,An_TCM_2017}. In C-RAN, a central baseband unit (BBU) is in charge of resource allocation and signal processing, while the low-cost and low-power remote radio heads (RRHs) are connected to the BBU via wireless or wired fronthaul links~\cite{Hao_IoT_2019}. The capacity of the fronthaul links becomes bottleneck for C-RAN, causing severe latency and degraded user experience, especially for applications such as virtual reality (VR), UHDV, and autonomous driving~\cite{Vu_TCOM_2018}. To tackle this problem, edge cache  has been developed to relieve the burden of the fronthaul links and lower the end-to-end latency~\cite{Yao_IoT_2019}. The main idea of edge cache is that RRHs should pre-fetch  the most frequently requested files from the cloud and store them in the RRHs' local caches during the off-peak traffic periods (such as midnight)~\cite{Le_TWC_2018}. When the files required  by the actuators (devices) are cached in the RRHs, they can be directly transmitted to the actuators from the RRHs, leading to reduced fronthaul link data traffic and decreased latency~\cite{He_TCOM_2019}. Therefore, the cache technique will play a pivotal role in future wireless networks~\cite{Hao_SJ_2019}. 

Edge cache-assisted C-RAN enables cooperation among RRHs to cancel the inter-RRH interference via coordinated multiple-point transmission (CoMP)~\cite{Tanno_TWC_2010}.  However, the CoMP approach requires cooperative RRHs to share the actuators' data, so that the fronthaul links have to carry each actuator's data multiple times. This limits the cooperative size and increases the fronthaul burden. To address this problem, the multicast technique can be used to avoid multiple tranmission by multicasting the acutator's message to all cooperative RRHs at the same time. Dai~{\it{et al.}}~\cite{Yu_JSAC_2018} proposed to use the wireless multicast in a cache-assisted C-RAN, where the cooperative RRHs pre-store the popular contents via multicast fronthaul links. Hu~{\it{et al.}}~\cite{Hu_TWC_2017} also adopted the multicast beamforming approach over fronthaul links to deliver the actuator's data to a group of RRHs. 

Due to the broadcast nature of wireless transmission, confidential messages may be eavesdropped by malicious attackers, jeopardizing the secrecy of the information  transmission~\cite{Ng_TVT_2016,Chen_TVT_2018}. Traditionally, the security issues at the wireless communication have been handled at the higher layer by using encryption approaches. However, the huge growth in the number of wireless devices and the rapid development of computing technologies have surfaced the vulnerability of the conventional encryption methods~\cite{Yang_TCM_2015}. Thus, the physical layer security (PLS) has been developed as a complementary approach to the conventional encryption methods for securing confidential information~\cite{Wyner_1975}. Another important reason is that for low-end, low-energy IoT devices with limited battery life and computational capabilities, most of the available energy and resources should be dedicated to core application functionalities, and there may be little left for supporting security. PLS techniques allow light-weight encryption at the higher layer while ensuring the secrecy of transmissions. The key idea of PLS is to use the randomness of the wireless channels to refrain the illegitimate side from wiretapping the users' information~\cite{Chu_TWC_2016}.  In this paper, we study the edge cache-assisted secure low-latency transmissions. Due to its ultra wide bandwith, millimeter wave (mmWave) is applied to the fronthaul and access links~\cite{Hao_TWC_2018}.   Specifically, for a given cache strategy, the central processor (CP) first delivers the required non-cached files to cooperative RRHs via the mmWave multicast fronthaul links, and then the RRHs jointly transmit the overall files to the actuators. The system objective  is to design the beamforming of the CP and the RRHs to minimize the secure transmission delay.

\subsection{Related Works}
Edge caching has attracted increasing attention recently. Some woks focused on the delivery strategies~\cite{Tao_TWC_2016,Fu_WCL_2019,He_JSAC_2019,Liu_JSAC_2016} and others studied the cache placement problems~\cite{Liu_TWC_2017,Zheng_CL_2017,Zhu_TWC_2018,Cui_TCOM_2018}. Specifically,  in~\cite{Tao_TWC_2016}, Tao~{\it{et al.}} proposed to form a multicast group for users requesting the same content,  so that  these users can be jointly served by the same group of RRHs. Under a given cache strategy, the authors investigated the dynamic RRH clustering and multicast beamforming to minimize the weighted sum of backhaul cost and transmit power. Fu~{\it{et al.}}~\cite{Fu_WCL_2019} studied the power control problem for non-orthogonal multiple access (NOMA) transmissions in wireless cache networks. The authors proposed a deep neural network-based method to minimize the transmission delay. In a cache-enabled multigroup multicasting network, He~{\it{et al.}}~\cite{He_JSAC_2019} designed three transmission schemes to minimize the delivery latency.  In~\cite{Liu_JSAC_2016}, Liu~{\it{et al.}} explored the potential of energy efficiency (EE) of the cache-enabled networks, and identified the optimal cache capacity that maximizes the EE. Furthermore, the authors analyzed the obtained EE gain brought by cache. Based on the flexible physical-layer transmission and the diverse requirements of different users,  Liu~{\it{et al.}}~\cite{Liu_TWC_2017} studied the cache placement problem, and proposed centralized and distributed cache strategies to minimize the download delay.  In~\cite{Zheng_CL_2017}, Zheng~{\it{et al.}} applied the cache technique to a distributed relay system, and proposed a hybrid cache scheme to minimize the outage probability.  Zhu~{\it{et al.}}~\cite{Zhu_TWC_2018} investigated the performance of cache-enabled ultra-dense small cell networks, and derived the successful content delivery probability (SCDP). The authors proposed two algorithms, namely constrained cross-entropy algorithm and heuristic probabilistic content placement algorithm to minimize the SCDP.  Under random cache at the RRHs, Cui~{\it{et al.}}~\cite{Cui_TCOM_2018} proposed two cooperative transmission schemes by jointly considering RRH cache and cooperation. The authors derived the expression of the successful transmission probability based on each scheme. 

Although the cache problem has been investigated  in~\cite{Tao_TWC_2016,Fu_WCL_2019,He_JSAC_2019,Liu_JSAC_2016,Liu_TWC_2017,Zheng_CL_2017,Zhu_TWC_2018,Cui_TCOM_2018}, the secure transmission aspect was not considered.
In~\cite{Ng_TWC_2018}, Derrick~{\it{et al.}} designed a cache scheme that effectively enhances the PLS for a backhaul-limited cellular network. However, since the authors' focus was to minimize the transmit power subject to the secrecy rate, the secure transmission delay  has not been considered.  In~\cite{Xu_WCL_2019}, Xu~{\it{et al.}} investigated the PLS in a mobile edge computing network, and formulated a weighted sum energy minimization problem under a given secure transmission delay. Wang~{\it{et al.}}~\cite{Wang_TCOM_2018} studied the security transmission problem in a cache-assisted heterogeneous network. To realize the secure and energy-efficient transmissions, a joint cache placement and file delivery scheme was proposed. In~\cite{Cheng_TCOM_2019}, Cheng~{\it{et al.}} considered a cache-assisted unmanned aerial vehicle (UAV) system, and developed a joint optimization strategy via designing UAV trajectory and time scheduling to improve the secure transmission.  Kiskani~{\it{et al.}}~\cite{Kiskani_TVT_2017} investigated the secure approach in an Ad Hoc network with cache. The authors proposed a novel decentralized secure coded caching scheme to enhance the secure storage, where the nodes only transmit the coded file for protecting the user information. Most aforementioned  works mainly focus on the cache placement design to enhance the PLS without considering the secure transmission delay. 

\begin{figure}[t]
	\begin{center}
		\includegraphics[width=6.5cm,height=4.5cm]{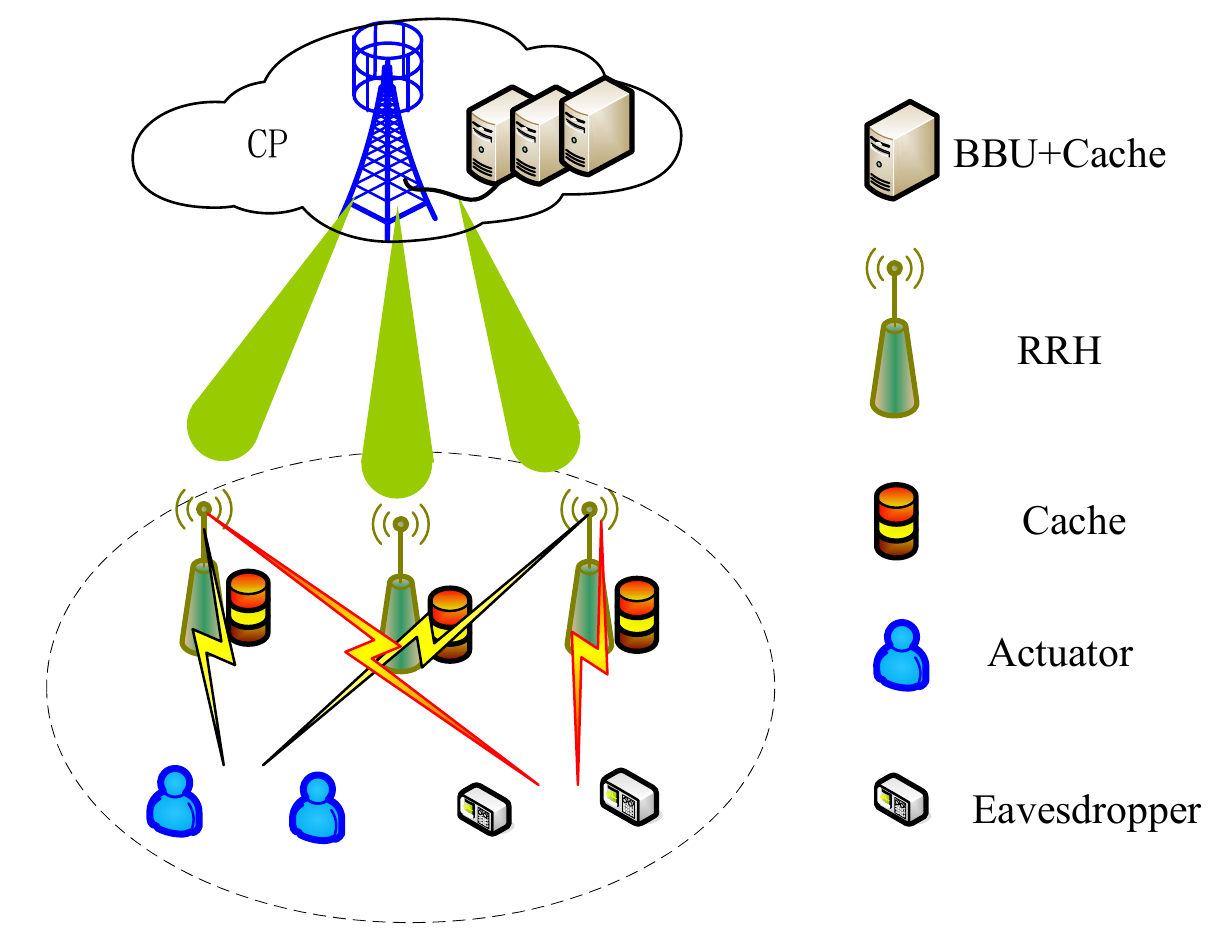}
		\caption{System model of the edge cache-assisted C-RAN with multicast fronthaul.}
		\label{systemfigure1}
	\end{center}
\end{figure}

\subsection{Main Contributions} 
In this paper, we investigate the secure transmission delay minimization problem in an edge cache-assisted mmWave C-RAN, where each RRH is equipped with a local cache. To reduce the hardware cost and energy consumption, we consider the single radio frequency (RF) chain and multiple antennas structure at the CP and RRHs. In addition, the CoMP technique is adopted among the RRHs. Therefore, the main challenge is how to jointly design the transmit and receive beamforming at the first phase as well as the cooperative beamforming for secure transmission at the second phase,  and the main contributions of this paper include:
\begin{itemize}
\item We develop a two-phase transmission frame structure. At the first phase, the uncached files are fetched from the CP to RRHs through the multicast fronthaul link. At the second phase, all required files are transmitted from the  RRHs to the actuators via the CoMP technique. In this regard, we formulate a secure transmission delay minimization problem by jointly optimizing transmit and receive beamforming at the first phase as well as the cooperative beamforming among RRHs at the second phase. 
\item Since the original problem is intractable, we divide it into two independent optimization problems, and minimize the transmission delay for each problem. For the former, we need to jointly optimize the transmit beamforming at the CP and the receive beamforming at the RRHs, and an alternatively iterative algorithm is proposed.
\item For the second phase, we minimize the maximum secure transmission delay to guarantee fairness. Moreover, the general scenario  with imperfect CSIs for eavesdroppers (Eves) links are assumed. We first design the analog beamforming for each RRH. Then, by the successive convex approximation (SCA) and $S$-procedure techniques, the problem is transformed into a  semidefinite programm (SDP), which can be recast into a convex one by dropping the rank-one constraint. Meanwhile, our test results show that semidefinite relaxation (SDR) gives rank-one solutions with nearly 99\% probability.
\end{itemize}

The rest of this paper is organized as follows. The system model and problem formulation are presented in Section~II. In Section III, the solution to the formulated secure transmission delay problem is provided. Simulation results are presented in Section~IV. Finally, conclusions are drawn  in Section~V.

\textit{Notations}: We use the following notations throughout this paper: $(\cdot)^\ast$, 
$(\cdot)^T$ and $(\cdot)^H$ denote the conjugate, transpose and Hermitian transpose, respectively, $\|\cdot\|$ is the Euclidean norm, ${\mathbb{C}}^{x\times y}$ means the space of $x\times y$ complex matrix, {Re($\cdot$)} and {Tr($\cdot$)} denote real number operation and trace operation, respectively. $[\cdot]^+$ denotes the $\max\{0,\cdot\}$, and Diag(${\bf{f}}_1,\ldots,{\bf{f}}_M$) is a diagonal matrix.

\section{System Model and Problem Formulation}
In this section, we first describe the  studied system model, and propose a two-phase frame structure for the file transmission. Next, we define the secure transmission delay and formulate a transmission delay minimization problem.
\subsection{System Model}

\begin{figure}[t]
	\begin{center}
		\includegraphics[width=8cm,height=4cm]{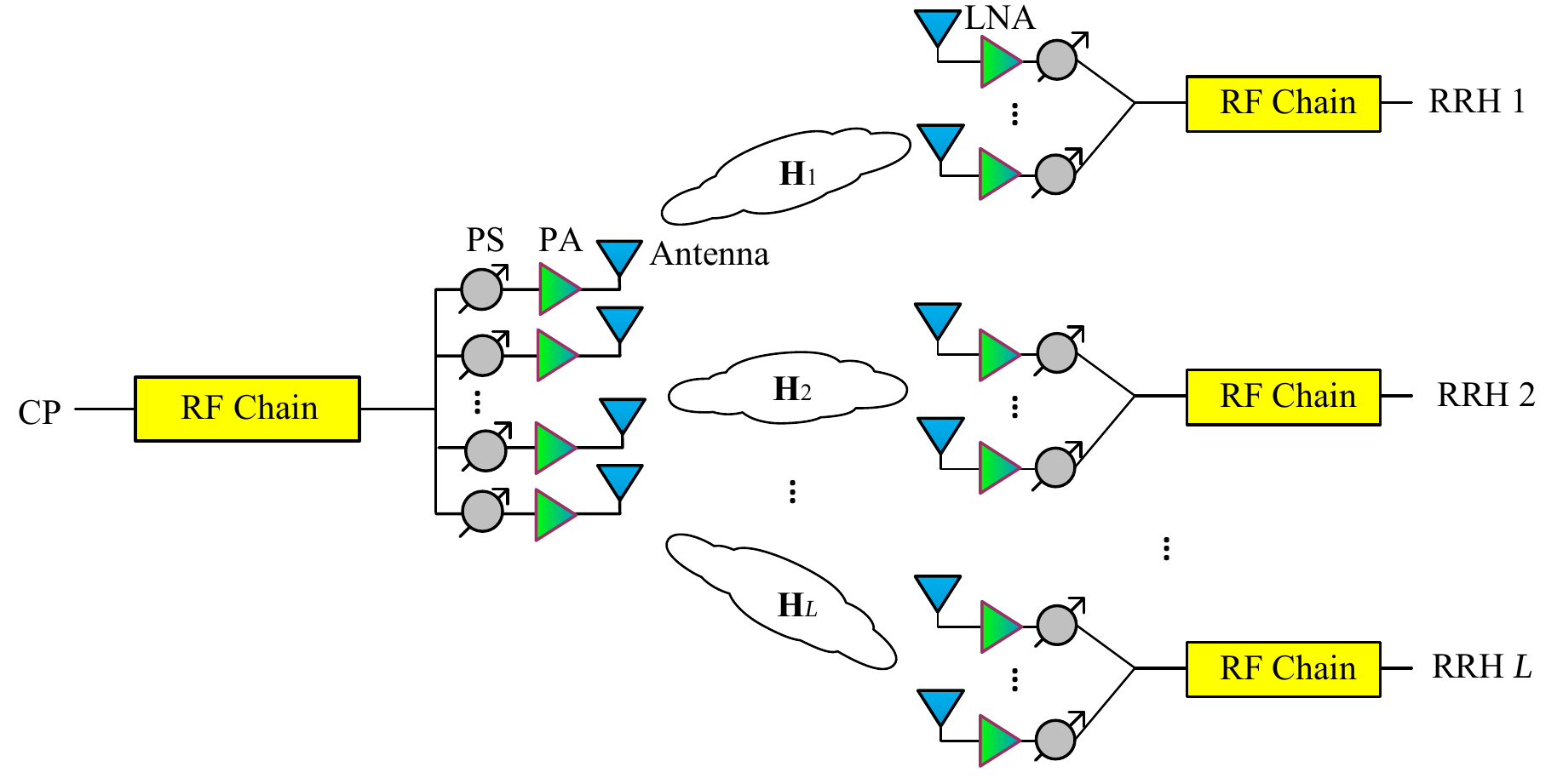}
		\caption{Single RF chain and multiple antenna structure for the CP and RRHs.}
		\label{systemfigure2}
	\end{center}
\end{figure}

\begin{figure}[t]
	\begin{center}
		\includegraphics[width=8cm,height=3.2cm]{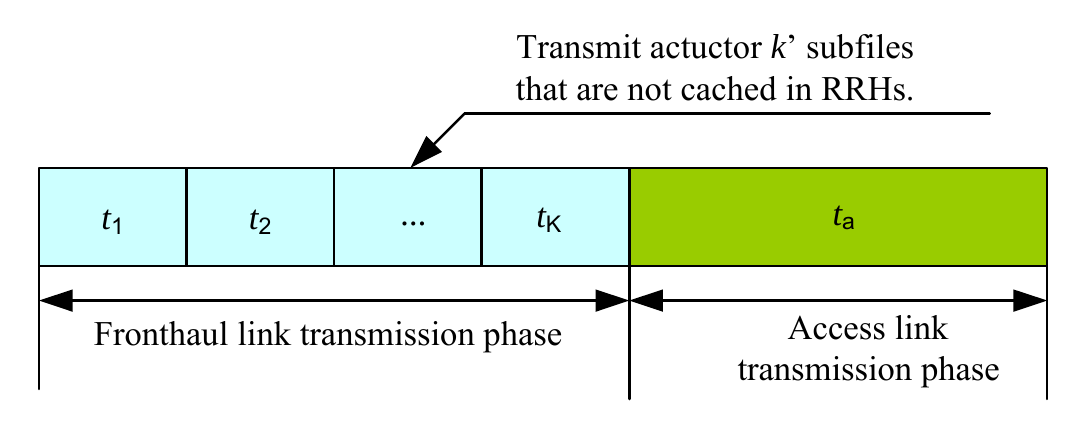}
		\caption{The transmission frame structure.}
		\label{systemfigure3}
	\end{center}
\end{figure}

\begin{table}[t]
	\caption{{List of Key Notations.}}
	\begin{center}
		\begin{tabular}{|c|c|}
			\hline
			$L$, $\mathcal{L}$ & Number and set of RRHs \\\hline
		    $K$, $\mathcal{K}$ & Number and set of actuators \\\hline
		    $S$, $\mathcal{S}$ & Number and set of Eves \\\hline
		    $M$, $\mathcal{M}$ & Number and set of antennas at the CP \\\hline
		    $N$, $\mathcal{N}$ & Number and set of antennas at each RRH \\\hline
		    $F$, $\mathcal{F}$ & Number and set of files \\\hline
		    $U$, $\mathcal{U}$ & Number and set of segments \\\hline
		    $L$, $\mathcal{L}$ & Number and set of RRHs \\\hline
		    $b_{f,u}$ & The cache state of segment ($f,u$) \\\hline
		    $c_{k,f}$ & If actuator $k$ requires file $f$ \\\hline
		    ${\bf{H}}_l$& Downlink channel from the CP to RRH $l$\\\hline
		    ${\bf{g}}_k$& Downlink channel from  $L$ RRHs to actuator $k$\\\hline
		    ${\bf{g}}_s^e$& Downlink channel from  $L$ RRHs to eavesdropper $s$\\\hline
		    ${\bf{w}}$& Transmit beamforming vector of the CP\\\hline
		    ${\bf{q}}_l$& Receive beamforming vector of RRH $l$\\\hline
		    ${\bf{z}}_l$& Transmit beamforming vector of RRH $l$\\\hline
		    ${\bf{v}}$& Cooperative digital precoding vector for $L$ RRHs\\\hline
		     $P$& Transmit power of the CP\\\hline
		     $P_{\rm{max}}^l$& Maximum transmit power of RRH $l$\\\hline
		    $t_k$& Delivery delay for transmitting actuator $k$' file from the CP\\\hline
		    
		\end{tabular}
	\end{center}
\end{table} 

As shown in Fig.~\ref{systemfigure1}, we consider the downlink transmission of a cache-enabled C-RAN, which includes one CP, $L$ RRHs, $K$ actuators. Meanwhile, there are $S$ Eves that may intercept the actuators' information. The same frequency mmWave  carrier is adopted at the fronthaul links from the CP to the RRHs and the access links from the RRHs to the actuators.  To reduce the energy consumption and hardware cost, the CP and RRHs are all equipped with a single RF chain, which is connected to multiple antennas via phase shifters (PSs) and power amplifier (PA) or low noise amplifier (LNA), as illustrated in Fig.~\ref{systemfigure2}. Denote the antenna number at the CP and each RRH by $M$ and $N$, respectively. In the considered system, each RRH $l\in\mathcal{L}=\{1,\dots,L\}$ is equipped with a finite-size cache, and can pre-fetch files from the CP during the off-peak period based on the content popularity and predefined cache strategies~\cite{Liu_TWC_2017,Hu_TVT_2019}. In addition, we assume that the BBU stores all the files requested by the actuators in its library. Denote the file set by $\mathcal{F}=\{1,\dots,F\}$, and each actuator requests only a file in a given time interval. In addition, each file $f$ can be split into $U$ segments with equal size, and each segment $(f,u)\;u\in\{1,\dots,U\}$ can be independently cached at the RRHs.  

In this paper, the CoMP technique is adopted, where multiple RRHs form one cluster and cooperatively serve the actuators. Here, we only consider one cluster to facilitate the analysis,  but our proposed scheme can be readily extended to multiple clusters. Similar to~\cite{He_JSAC_2019,Ng_TWC_2018}, we denote the cache status of segment $(f,u)$ in the cooperative RRH cluster by $b_{f,u}$, which can be expressed as 
\begin{eqnarray}\label{cache1}
b_{f,u}=\left\{
\begin{aligned}
&1,\;\; {\rm{if\;segment}}\; (f,u)\;{\rm{is\;cached\;in\;RRHs}},\\
&0,\;\;{\rm{otherwise}}.
\end{aligned}
\right.
\end{eqnarray}
Here, we propose a transmission frame structure as shown in Fig.~\ref{systemfigure3}, which includes two transmission phases: the fronthaul transmission phase and the access transmission phase.  During the fronthaul transmission phase, the CP delivers the uncached files requested by the actuators to RRHs via multicast, where $t_k$ denotes the delivery time for transmitting actuator $k$'s file. When all required files are fetched from the CP, $L$ RRHs cooperatively serve $K$ actuators during the access transmission phase.  

During the fronthaul link transmission phase,  the received signal by the RRH $l$ can be expressed as
\begin{eqnarray}\label{Fronth1}
	y_l={\bf{q}}_l{\bf{H}}_l{\bf{w}}\sqrt{P}x+{\bf{q}}_l{\bf{n}}_l,
\end{eqnarray}
where ${\bf{H}}_l\in\mathbb{C}^{N\times M}$ denotes the downlink mmWave channel matrix from the CP to RRH $l$, ${\bf{q}}_l\in\mathbb{C}^{1\times N}$ and ${\bf{w}}\in\mathbb{C}^{M\times 1}$ denote the transmit beamforming vector of the CP and the receive beamforming vector of RRH $l$, respectively;  $x$ denotes the multicast signal, satisfying $\mathbb{E}\{|x|^2\}=1$, and $P$ denotes the transmit power of the CP;  ${\bf{n}}_l$ is the independent and identically distributed (i.i.d.) additive white Gaussian noise (AWGN) vector, where each entry follows  $\mathcal{CN}(0,\delta^2)$. Due to the constant amplitude modulation for the transmit and receive beamforming, we have $|[{\bf{q}}_l]_n|=1/\sqrt{N}\;(n\in\mathcal{N})$ and  $|[{\bf{w}}]_m|=1/\sqrt{M}\;(m\in\mathcal{M})$, where $\mathcal{N}=\{1,\cdots,N\}$ and $\mathcal{M}=\{1,\cdots,M\}$~\cite{Hao_Access_2018}.

For the  mmWave channel, we adopt the widely used limited-scattering channel model with a uniform linear array~\cite{Xiao_TWC_2016}. Each scatter is assumed to contribute to a single propagation path. The mmWave channel can be expressed as
\begin{eqnarray}\label{channelmodel}
{\bf{H}}_l=\sum\nolimits_{g=1}^{G}\alpha_{l,g}{\bf{a}}_{\rm{RRH}}(\theta_{l,g}){\bf{a}}^H_{\rm{CP}}(\phi_{l,g}),
\end{eqnarray}
where $G$ is the number of multiple paths, $\alpha_{l,g}$ represents the complex gain of the path $g$ from the CP to RRH $l$, $\theta_{l,g}$ and $\phi_{l,g}$ are the angle of arrival  (AoA) and angle of departure (AoD) of path $g$, respectively.  When the half-wavelength antenna space is adopted,  the steering vectors   ${\bf{a}}_{\rm{RRH}}(\theta_{l,g})$ and ${\bf{a}}_{\rm{CP}}(\phi_{l,g})$ can be expressed as
	\begin{subequations}
		\begin{align}
 {\bf{a}}_{\rm{RRH}}(\theta_{l,g})=[1,e^{j\pi 1\theta_{l,g}},e^{j\pi 2\theta_{l,g}},\dots, e^{j\pi (N-1)\theta_{l,g}}]^T,\\
 {\bf{a}}_{\rm{CP}}(\phi_{l,g})=[1,e^{j\pi 1\phi_{l,g}},e^{j\pi 2\phi_{l,g}},\dots, e^{j\pi (M-1)\phi_{l,g}}]^T.
 \end{align}
\end{subequations} 
Based on~(\ref{Fronth1}), the achievable rate of RRH $l$ can be written as
\begin{eqnarray}\label{MulticastRate}
	R_l({\bf{q}}_l, {\bf{w}})=\log\left(1+{P|{\bf{q}}_l{\bf{H}}_l{\bf{w}}|^2}/{\sigma_l^2}\right).
\end{eqnarray} 

It is clear that the achievable multicast fronthaul rate is limited by the RRH with the worst channel condition, and is given by
\begin{eqnarray}
  R_1({\bf{Q}}, {\bf{w}})=\underset{l\in{\mathcal{L}}}{\min}\;\left\{R_l({\bf{q}}_l, {\bf{w}})\right\},
\end{eqnarray}  
where ${\bf{Q}}=[{\bf{q}}_1^T,\dots,{\bf{q}}_L^T]^T$. Accordingly, the fronthaul transmission delay can be calculated as
\begin{eqnarray}\label{FronDelay}
	T=\sum\nolimits_{k=1}^{K}t_k=\sum\nolimits_{k=1}^{K}\frac{\Psi_k}{R_1({\bf{Q}}, {\bf{w}})},
\end{eqnarray} 
where $\Psi_k=\sum_{f=1}^{F}c_{k,f}\frac{\Omega_f}{U}\sum_{u=1}^{U}(1-b_{f,u})$, $\Omega_f$ denotes the size of file~$l$, while $c_{k,f}$ is a binary variable, satisfying $\sum_{f=1}^F c_{k,f}=1$. Note that $c_{k,f}=1$ when actuator $k$ requires file $f$, and vice versa.

During the access transmission phase, $L$ RRHs cooperatively serve all actuators, and we denote the downlink channel vector from RHHs to actuator $k$ as ${\bf{g}}_k=[{\bf{g}}_k^1,{\bf{g}}_k^2,\dots,{\bf{g}}_k^L]$, where ${\bf{g}}_k^l\in{\mathbb{C}}^{1\times N}$ represents the donwnlink channel vector from RRH $l$ to actuator $k$. Denote the signal for actuator $k$ as $x_k$, satisfying $\mathbb{E}{|x_k|^2}=1$. The received signal by the actuator $k$ can thus be expressed as
\begin{eqnarray}
	y_k={\bf{g}}_k{\bf{Z}}{\bf{v}}_kx_k+\sum\nolimits_{i\neq k}^{K}{\bf{g}}_k{\bf{Z}}{\bf{v}}_ix_i+n_k,
\end{eqnarray}
where ${\bf{v}}_k\in {\mathbb{C}^{L\times 1}}$ is the digital beamformer for actuator $k$. ${\bf{Z}}\in\mathbb{C}^{NL\times L}$ denotes the   analog beamforming matrix, and is given as Diag(${\bf{z}}_{1},\ldots,{\bf{z}}_{L}$),
where ${\bf{z}}_{l}\in\mathbb{C}^{N\times 1}$ denotes the analog beamforming vector of RRH~$l$. The achievable rate of actuator $k$ can be written as 
\begin{eqnarray}\label{rate1}
   R_{2,k}({\bf{Z}},{\bf{V}})=\log\left(1+\frac{|{\bf{g}}_k{\bf{Z}}{\bf{v}}_k|^2}{\sum_{i\neq k}^K|{\bf{g}}_k{\bf{Z}}{\bf{v}}_i|^2+\sigma_k^2}\right),
\end{eqnarray} 
where ${\bf{V}}=[{\bf{v}}_1,\dots,{\bf{v}}_K]$.

We consider the scenario when the channel fading between the CP and the actuators or Eves is so large that the actuators and Eves cannot receive the information from the CP.  Therefore, the Eves only attempt to intercept the information of the actuators from the RRHs, and the received signal at Eve $s$ can be expressed as
\begin{eqnarray}
y_s={\bf{g}}_s^e{\bf{Z}}{\bf{v}}_kx_k+\sum\nolimits_{i\neq k}^{K}{\bf{g}}_s^e{\bf{Z}}{\bf{v}}_ix_i+n_s,
\end{eqnarray}
where ${\bf{g}}_s^e\in{\mathbb{C}}^{1\times NL}$ denotes the channel vector from the RRHs to Eve $s$. The mmWave channels ${\bf{g}}_k$ and ${\bf{g}}_s^e$ have the similar structure to~(\ref{channelmodel}), and thus, the detailed expressions are omitted here.  The intercepted rate of Eve $s$ on actuator $k$'s signal can be written~as
\begin{eqnarray}\label{rate2}
R_{2,k}^s({\bf{Z}},{\bf{V}})=\log\left(1+\frac{|{\bf{g}}_s^e{\bf{Z}}{\bf{v}}_k|^2}{\sum_{i\neq k}^K|{\bf{g}}_s^e{\bf{Z}}{\bf{v}}_i|^2+\sigma_s^2}\right).
\end{eqnarray} 

Finally, the achievable secrecy rate by  actuator $k$ can be expressed as~\cite{Zeng_JSTSP_2019,Chu_TVT_2016}
\begin{eqnarray}\label{securerate}
	\hat{R}_{2,k}({\bf{Z}},{\bf{V}})=\left[R_{2,k}({\bf{Z}},{\bf{V}})-\underset{s\in{\mathcal{S}}}{\max}\{R_{2,k}^s({\bf{Z}},{\bf{V}})\}\right]^+,
\end{eqnarray}
where $\mathcal{S}=\{1,2,\dots,S\}$ denotes the Eve set. Next, we define the secure transmission delay for actuator $k$ as
\begin{eqnarray}\label{SecureDelay}
	T_k=\frac{\hat{\Psi}_k}{\hat{R}_{2,k}({\bf{Z}},{\bf{V}})}.
\end{eqnarray}
where $\hat{\Psi}_k=\sum_{f=1}^{F}c_{k,f}{\Omega_f}$.
\subsection{Problem Formulation}
In general,  cache status $b_{f,u}$ or cache placement depends on several factors, such as user behavior, information popularity distribution and so on, and it can be decided by several advanced cache schemes, e.g., artificial intelligence-based multi-timescale framework method~\cite{Hu_TVT_2019} and deep $Q$-learning method~\cite{Hu_TVT_2018}. Here, we assume that the cache status  $b_{f,u}$ has been fixed according to a certain cache strategy, and the required files by the actuators, i.e., $c_{k,f}$ are also given in advance~\cite{Tao_TWC_2016,He_JSAC_2019}. In this work, we mainly focus on the beamforming design to minimize the two-phase transmission delay. 

To ensure system fairness, we aim to minimize the maximum secure transmission delay,  and formulate the following optimization problem:
\begin{subequations}\label{OptA}
	\begin{align}
	\;\;\;\;\;\;&\underset{\left\{{\bf{Q}}, {\bf{{w}}}, {\bf{Z}}, {\bf{V}}\right\}}{\rm{min}}\;\;\underset{k\in{\mathcal{K}}}{\rm{max}}\;{T_k}+\sum\nolimits_{k=1}^{K}t_k\;\;\label{OptA0}\\
	{\rm{s.t.}}\;\; &\sum\nolimits_{k=1}^{K}|{\bf{v}}_k(l)|^2\leq P_{\rm{max}}^l, l\in \mathcal{L},\label{OptA1}\\
	&|[{\bf{q}}_l]_n|={1}/{\sqrt{N}},\;\;|[{\bf{z}}_{l}]_n|={1}/{\sqrt{N}}, n\in \mathcal{N}, l\in \mathcal{L},\label{OptA2}\\
	&|[{\bf{w}}]_m|={1}/{\sqrt{M}},m\in \mathcal{M},\label{OptA3}
	\end{align}
\end{subequations}
where (\ref{OptA1}) denotes per-RRH transmit power constraint.  To handle problem (\ref{OptA}), we need to design the analog beamformers for both the CP and the RRHs at the fronthaul transmission phase as well as the hybrid analog/digital beamformer of cooperative RRHs at the access transmission phase, which is an intractable problem. 
\section{Problem Solution}  
Since the two transmission phases are relatively independent, we can equivalently divide the original problem into two parts, i.e., the transmit and receive beamforming design problem at the fronthaul transmission phase (i.e., ${\mathbb{P}1}$) as well as the hybrid analog/digital beamforming design problem at the access transmission phase (i.e., ${\mathbb{P}2}$), namely:
\begin{subequations}\label{OptB}
	\begin{align}
	{\mathbb{P}1}\;\;\;\;\;\;\;\;\;\;&\underset{\left\{{\bf{Q}}, {\bf{{w}}}\right\}}{\rm{min}}\;\;\sum\nolimits_{k=1}^{K}t_k\;\;\label{OptB0}\\
	{\rm{s.t.}}\;\; &|[{\bf{q}}_l]_n|={1}/{\sqrt{N}}, n\in \mathcal{N}, l\in \mathcal{L},\label{OptB2}\\
	&|[{\bf{w}}]_m|={1}/{\sqrt{M}},m\in \mathcal{M}.\label{OptB3}
	\end{align}
\end{subequations}
\begin{subequations}\label{OptC}
	\setlength{\abovedisplayskip}{-10pt}
	\begin{align}
	{\mathbb{P}2}\;\;\;\;\;\;\;\;\;\;&\underset{\left\{{\bf{Z}}, {\bf{V}}\right\}}{\rm{min}}\;\;\underset{k\in {\mathcal{K}}}{\rm{max}}\;{T_k}\;\;\label{OptC0}\\
	{\rm{s.t.}}\;\; &\sum\nolimits_{k=1}^{K}|{\bf{v}}_k(l)|^2\leq P_{\rm{max}}^l, l\in \mathcal{L},\label{OptC1}\\
	&|[{\bf{z}}_{l}]_n|={1}/{\sqrt{N}}, n\in \mathcal{N}, l\in \mathcal{L}.\label{OptC2}
	\end{align}
\end{subequations} 
\subsection{Beamforming Design at the Fronthaul Link}
Due to  $\sum_{k=1}^{K}t_k=\frac{1}{{R_1({\bf{Q}}, {\bf{w}})}}\sum_{k=1}^{K}\Psi_k$, we can rewrite ${\mathbb{P}1}$ as the following max-min rate problem
\begin{subequations}\label{OptD}
	\begin{align}
	\;\;\;\;\;\;\;\;\;\;\;\;\;\;\;\;\;\;\;\;&\underset{\left\{{\bf{Q}}, {\bf{{w}}}\right\}}{\rm{max}}\;\;\underset{l\in{\mathcal{L}}}{\min}\;R_l({\bf{q}}_l, {\bf{w}})\;\;\label{OptD0}\\
	{\rm{s.t.}}\;\; &{\rm{(\ref{OptB2}),(\ref{OptB3})}}.\label{OptD1}
	\end{align}
\end{subequations}
Based on~(\ref{MulticastRate}), (\ref{OptD}) can be equivalently written as follows:
\begin{subequations}\label{OptDD}
	\begin{align}
	\;\;\;\;\;\;\;\;\;\;\;\;\;\;\;\;\;\;\;\;&\underset{\left\{{\bf{Q}}, {\bf{{w}}}\right\}}{\rm{max}}\;\;\underset{l\in{\mathcal{L}}}{\min}\;|{\bf{q}}_l{\bf{H}}_l{\bf{w}}|^2\;\;\label{OptDD0}\\
	{\rm{s.t.}}\;\; &{\rm{(\ref{OptB2}),(\ref{OptB3})}}.\label{OptDD1}
	\end{align}
\end{subequations}

(\ref{OptDD}) is still difficult to handle due to the non-smooth objective function and the non-convex constraints. Moreover, the multiplication among the optimization variables makes it even more intractable.  To this end, we first initialize the receive beamformer  $\hat{\bf{Q}}=[\hat{\bf{q}}_1^T,\dots,\hat{\bf{q}}_L^T]^T$. By introducing the auxiliary variable $\eta$ and relaxing the constraint (\ref{OptB2}), (\ref{OptDD}) can be reformulated as the following tractable optimization problem
\begin{subequations}\label{OptDE}
	\setlength{\abovedisplayskip}{-6pt}
	\begin{align}
	\;\;\;\;\;\;\;\;\;\;\;\;\;\;\;\;\;\;\;\;&\underset{\left\{{\bf{{w}}},\eta\right\}}{\rm{max}}\;\;\eta\;\;\label{OptDE0}\\
	{\rm{s.t.}}\;\; &|\hat{\bf{h}}_l{\bf{w}}|^2\geq \eta, m\in \mathcal{M}, \label{OptDE1}\\
	&|[{\bf{w}}]_m|\leq {1}/{\sqrt{M}},m\in \mathcal{M}, \label{OptDE2}
	\end{align}
\end{subequations}  
where $\hat{\bf{h}}_l=\hat{\bf{q}}_l{\bf{H}}_l$.
\begin{algorithm}[t]
	\setlength{\belowdisplayskip}{-6pt}
	{\caption{The Proposed Alternatively Iterative Algorithm.}
		\label{algorithm1}
		{\bf{Initialize}} $\hat{\bf{q}}_l$ for any $l$.\\
		\Repeat{{\rm{Convergence}}}{
			Compute transmit beamformer  ${\bf{w}}^\star$ according to~(\ref{tr});\\
			Compute receive beamformer  ${\bf{q}}_l^\star$ according to~(\ref{re});\\
			Update $\hat{\bf{q}}_l={\bf{q}}_l^\star$ for any $l$;
		}
	}
\end{algorithm} 
One can observe that (\ref{OptDE1}) is the only non-convex constraint. According to the first-order Taylor approximation formula, $|\hat{\bf{h}}_l{\bf{w}}|^2$ can be approximated as
\begin{eqnarray}\label{Taylor}
	|\hat{\bf{h}}_l{\bf{w}}|^2\approx \hat{\bf{w}}^H\hat{\bf{H}}_l\hat{\bf{w}}+2{\rm{Re}}\{\hat{\bf{w}}^H\hat{\bf{H}}_l({\bf{w}}-\hat{\bf{w}})\},
\end{eqnarray}
where $\hat{\bf{H}}_l=\hat{\bf{h}}_l^H\hat{\bf{h}}_l$, and $\hat{\bf{w}}$ denotes the initial transmit beamformer. Finally, we formulate the convex optimization problem~as 
\begin{subequations}\label{OptDG}
	\setlength{\abovedisplayskip}{-6pt}
	\begin{align}
	\;\;\;\;\;&\underset{\left\{{\bf{{w}}},\eta\right\}}{\rm{max}}\;\;\eta\;\;\label{OptDG0}\\
	{\rm{s.t.}}\;\; &\hat{\bf{w}}^H\hat{\bf{H}}_l\hat{\bf{w}}+2{\rm{Re}}\{\hat{\bf{w}}^H\hat{\bf{H}}_l({\bf{w}}-\hat{\bf{w}})\}\geq \eta, m\in \mathcal{M}, \label{OptDG1}\\
	&|[{\bf{w}}]_m|\leq {1}/{\sqrt{M}},m\in \mathcal{M}. \label{OptDG2}
	\end{align}
\end{subequations}  
The above problem can be solved using the standard convex optimization techniques, e.g., interior-point method. However, the obtained solutions may not satisfy the constraint (\ref{OptB3}). To address this issue, we normalize each element of ${\bf{{w}}}^\star$ as follows
\begin{eqnarray}\label{tr}
	[{\bf{w}}^\star]_m=\frac{1}{\sqrt{N}}\frac{[{\bf{w}}^\star]_m}{||[{\bf{w}}^\star]_m||},m\in \mathcal{M},
\end{eqnarray}
where ${\bf{w}}^\star$ denotes the solution of~(\ref{OptDG}). Upon obtaining ${\bf{w}}^\star$, we proceed with the design of the receive beamformer for each RRH, namely
\begin{subequations}\label{OptDF}
	\begin{align}
	\;\;\;\;\;\;\;\;\;\;\;\;\;\;\;\;\;\;\;\;&\underset{\left\{{\bf{q}}_l\right\}}{\rm{max}}\;\;|{\bf{q}}_l{\bf{H}}_l{\bf{w}}^\star|^2\;\;\label{OptDF0}\\
	{\rm{s.t.}}\;\; &|[{\bf{q}}_l]_n|={1}/{\sqrt{N}}, n\in \mathcal{N}.\label{OptBF1}
	\end{align}
\end{subequations}

It is readily known that the optimal receive beamformer can be obtained as
\begin{eqnarray}\label{re}
	[{\bf{q}}_l^\star]_n=\frac{1}{\sqrt{N}}\frac{[{\bf{H}}_l{\bf{w}}^\star]_n^\ast}{||[{\bf{H}}_l{\bf{w}}^\star]_n||}, n\in \mathcal{N}.
\end{eqnarray}
Next, we replace $\hat{\bf{q}}_l$ in (\ref{OptDE}) with ${\bf{q}}_l^\star$ and resolve (\ref{OptDE}). The above process is repeated until the result converges. We summarize the alternatively iterative scheme in~Algorithm~\ref{algorithm1}. 

\subsection{Beamforming Design at the Access Link}
${\mathbb{P}}2$ is also difficult to handle due to the non-smooth and non-convex objective function (\ref{OptC0}) and the constant modulus constraint (\ref{OptC2}). Furthermore, joint optimization of analog and digital beamforming is extremely challenging.  Next, we first design analog beamformer ${\bf{Z}}$, and rewrite ${\bf{g}}_k{\bf{Z}}{\bf{v}}_k$ as
\begin{eqnarray}
	{\bf{g}}_k{\bf{Z}}{\bf{v}}_k=\sum\nolimits_{l=1}^{L}{\bf{g}}_k^l{\bf{z}}_l{\bf{v}}_k,
\end{eqnarray} 
where ${\bf{g}}_k=[{\bf{g}}_k^1,{\bf{g}}_k^2,\dots,{\bf{g}}_k^L]$ with ${\bf{g}}_k^l\in{\mathcal{C}}^{1\times N}$ denoting the sub-channel vector from RRH $l$ to actuator $k$. Similar to~\cite{Li_TWC_2016}, we design the sub-beamforming to maximize the equivalent sub-channel gain $|{\bf{g}}_k^l{\bf{z}}_l|^2$. The optimal sub-beamforming  can be expressed as
\begin{eqnarray}
	\left[{\bf{z}}_l\right]_n=\frac{1}{\sqrt{N}}\frac{[{\bf{g}}_k^{l}]_n^\ast}{||[{\bf{g}}_k^{l}]_n||}, n\in \mathcal{N}.
\end{eqnarray}

In addition, to guarantee fairness among the actuators, we maximize the equivalent sub-channel gain for each actuator in turn. Finally, we can obtain the analog beamformer ${\bf{Z}}^\star$ and the equivalent channel gain $\bar{\bf{g}}_k={\bf{g}}_k{\bf{Z}}^\star$. To this end, (\ref{rate1}) and (\ref{rate2}) can be rewritten as 
\begin{subequations}
	\begin{align}
\;\;\;\;\;\;\;\;\;\;\;R_{2,k}({\bf{V}})&=\log\left(1+\frac{|\bar{\bf{g}}_k{\bf{v}}_k|^2}{\sum_{i\neq k}^K|\bar{\bf{g}}_k{\bf{v}}_i|^2+\sigma_k^2}\right),\label{rate3}\\
R_{2,k}^s({\bf{V}})&=\log\left(1+\frac{|\bar{\bf{g}}_s^e{\bf{v}}_k|^2}{\sum_{i\neq k}^K|\bar{\bf{g}}_s^e{\bf{v}}_i|^2+\sigma_s^2}\right), \label{rate4}
\end{align}
\end{subequations} 
where $\bar{\bf{g}}_k^e={\bf{g}}_k^e{\bf{Z}}^\star$. 

In fact, Eves are usually passive, and their CSIs may not be perfectly known by the CP~\cite{Chu_TVT_20162}. Therefore, channel uncertainty is unavoidable and should be considered. In this paper,  we define the channel uncertainty as follows
\begin{eqnarray}
	\bar{\bf{g}}_s^e=\hat{\bf{g}}_s^e+\Delta \hat{\bf{g}}_s^e,
\end{eqnarray} 
where $\hat{\bf{g}}_s^e$ denotes the estimated equivalent channel vector, and $\Delta \hat{\bf{g}}_s^e$ is the corresponding error, which is assumed to be bounded by $\tau_s$, namely $\Delta \hat{\bf{g}}_s^e(\Delta \hat{\bf{g}}_s^e)^H\leq \tau_s$.

Next, we introduce an auxiliary variable $t$ and transform the original problem ${\mathbb{P}2}$ into
\setlength{\mathindent}{0cm}
\begin{subequations}\label{OptF}
	\begin{align}
	&\underset{\left\{{\bf{V}}\right\}}{\rm{min}}\;\;t\;\;\label{OptF0}\\
	{\rm{s.t.}}\;\;& R_{2,k}({\bf{V}})-\underset{s\in{\mathcal{S}}}{\max}\left\{R_{2,k}^s({\bf{V}})\right\}\geq \hat{\Psi}_k/{t}, k\in \mathcal{K},\label{OptF1}\\
	&\bar{\bf{g}}_s^e=\hat{\bf{g}}_s^e+\Delta \hat{\bf{g}}_s^e, \Delta \hat{\bf{g}}_s^e(\Delta \hat{\bf{g}}_s^e)^H\leq \tau_s,s\in{\mathcal{S}},\label{OptF2}\\
	&\sum\nolimits_{k=1}^{K}|{\bf{v}}_k(l)|^2\leq P_{\rm{max}}^l, l\in \mathcal{L}.\label{OptF3}
	\end{align}
\end{subequations} 

Problem (\ref{OptF}) is non-convex due to the constraints (\ref{OptF1}) and~(\ref{OptF2}). Next, we propose advanced approximation approaches to transform them into convex ones. By introducing auxiliary variables $\alpha_{k}$ and $\beta_{k}$, (\ref{OptF1}) can be split into the following constraints
\begin{subequations}
	\begin{align}
	\;\;\;\;\;\;\;\;\;\;\;\;\log(1+\alpha_{k})-\log(1+\beta_k)\geq \hat{\Psi}_k/{t}, k\in \mathcal{K},\label{A1}\\
	\alpha_{k}\leq \frac{|\bar{\bf{g}}_k{\bf{v}}_k|^2}{\sum_{i\neq k}^K|\bar{\bf{g}}_k{\bf{v}}_i|^2+\sigma_k^2}, k\in \mathcal{K},\label{B1}\\
	\beta_k \geq \frac{|\bar{\bf{g}}_s^e{\bf{v}}_k|^2}{\sum_{i\neq k}^K|\bar{\bf{g}}_s^e{\bf{v}}_i|^2+\sigma_s^2}, k\in \mathcal{K},s\in\mathcal{S}.\label{C1}
\end{align}
\end{subequations}
Nonetheless, (\ref{A1})-(\ref{C1}) are still non-convex constraints. Next, we approximate $\log(1+\beta_k)$ using the first order Taylor approximation formula, and obtain
\begin{eqnarray}
	\log(1+\beta_{k})\approx\log(1+\beta_{k}^{[i]})+\frac{\beta_{k}-\beta_{k}^{[i]}}{1+\beta_{k}^{[i]}},
\end{eqnarray}  
where $\beta_{k}^{[i]}$ denotes the value of $\beta_{k}$ at the $i$th iteration. On this basis, (\ref{A1}) can be transformed into the following convex constraint
\begin{eqnarray}\label{E1}
	\log(1\!+\!\alpha_{k})\!-\!\log(1\!+\!\beta_{k}^{[i]})\!-\!\frac{\beta_{k}\!-\!\beta_{k}^{[i]}}{1\!+\!\beta_{k}^{[i]}}\!\geq\! \hat{\Psi}_k/{t}, k\in \mathcal{K}.
\end{eqnarray}

To deal with~(\ref{B1}), we first define ${\bf{G}}_k={\bf{g}}_k^H{\bf{g}}_k$ and ${\bf{V}}_k={\bf{v}}_k{\bf{v}}_k^H$. By introducing auxiliary variable $\mu_k$, (\ref{B1}) can be split into the following constraints:
\begin{subequations}
	\begin{align}
	\;\;\;\;\;\;\;\;\;\;\;\;\;\alpha_{k}\mu_k&\leq {\rm{Tr}}(\bar{\bf{G}}_k{\bf{V}}_k),k\in \mathcal{K},\label{Eq1}\\
	\mu_k&\geq \sum\nolimits_{i\neq k}^K{{\rm{Tr}}(\bar{\bf{G}}_k{\bf{V}}_i)}+\sigma_k^2, k\in \mathcal{K}.\label{Eq2}
	\end{align}
\end{subequations} 

In fact, (\ref{Eq1}) and (\ref{Eq2}) can be regarded as SDP constraints with ${\rm{Rank}}({\bf{V}}_k)=1$. In addition, according to~\cite{Song_ICASSP_2016}, $\alpha_{k}\mu_k$ has the following upper bound as a valid surrogate function $\alpha_{k}\mu_k\leq \frac{\alpha_{k}^{[i]}}{2\mu_k^{[i]}}\mu_k^2+\frac{\mu_k^{[i]}}{2\alpha_{k}^{[i]}}\alpha_{k}^2$, and (\ref{Eq1}) can thus be transformed into the following convex constraint:
\begin{eqnarray}\label{Eq3}
 \frac{\alpha_{k}^{[i]}}{2\mu_k^{[i]}}\mu_k^2+\frac{\mu_k^{[i]}}{2\alpha_{k}^{[i]}}\alpha_{k}^2\leq {\rm{Tr}}(\bar{\bf{G}}_k{\bf{V}}_k),k\in \mathcal{K},
\end{eqnarray}
where $\alpha_{k}^{[i]}$ and $\mu_{k}^{[i]}$ represent the values of $\alpha_{k}$ and $\mu_{k}$ at the $i$th iteration, respectively.

To handle~(\ref{C1}), we introduce the classic $\mathcal{S}$-Procedure~\cite{boyd1994linear}:
\begin{lemma}
	Define the following function
	\begin{eqnarray}
	f_i({\bf{v}})={\bf{v}}{\bf{U}}_i{\bf{v}}^H+2{\rm{Re}}\{{\bf{c}}_i{\bf{v}}^H\}+b_i, i\in\{1,2\}, 
	\end{eqnarray}
	where ${\bf{v}}\in{\mathbb{C}}^{1\times \Gamma}$, ${\bf{U}}_i\in{\mathbb{C}}^{\Gamma\times \Gamma}$, ${\bf{c}}_i\in{\mathbb{C}}^{1\times \Gamma}$, $b_i\in\mathbb{R}$ and $\Gamma$ is any integer. If the following expression
	\begin{eqnarray}
		f_i({\bf{v}})\leq0 \Rightarrow  f_2({\bf{v}})\leq 0
	\end{eqnarray}
	holds, there must exist a $\lambda$ satisfying
	\begin{eqnarray}\label{C121}
	\lambda\left[ \begin{array}{ccc}
	{\bf{U}}_1 & {\bf{c}}_1^H \\
	{\bf{c}}_1 & b_1
	\end{array} 
	\right ]-\left[ \begin{array}{ccc}
	{\bf{U}}_2 & {\bf{c}}_2^H \\
	{\bf{c}}_2 & b_2
	\end{array} 
	\right ]\succeq {\bf{0}}.
	\end{eqnarray}  
\end{lemma}

Combing (\ref{OptF2}) and~(\ref{C1}), we have 
\begin{eqnarray}\label{Q1}
 \begin{aligned}
	&\Delta \hat{\bf{g}}_s^e {\bf{V}}_k (\Delta \hat{\bf{g}}_s^e)^H+2{\rm{Re}}\{\hat{\bf{g}}_s^e {\bf{V}}_k (\Delta \hat{\bf{g}}_s^e)^H\}+\hat{\bf{g}}_s^e {\bf{V}}_k (\hat{\bf{g}}_s^e)^H \\
	\leq&\beta_{k} \left(\sum_{i\neq k}^K\left(\Delta \hat{\bf{g}}_s^e {\bf{V}}_i (\Delta \hat{\bf{g}}_s^e)^H\!+\!2{\rm{Re}}\{\hat{\bf{g}}_s^e {\bf{V}}_i (\Delta \hat{\bf{g}}_s^e)^H\}\!+\!\hat{\bf{g}}_s^e {\bf{V}}_i (\hat{\bf{g}}_s^e)^H\right)\!+\!\sigma_s^2\right).
	\end{aligned}
\end{eqnarray}
Then, we introduce the auxiliary variables $\psi_k$, $\kappa_k$ and $\phi_k$, and split (\ref{Q1}) into the following constraints:
\setlength{\mathindent}{0cm}
\begin{subequations}
	\begin{align}
&\Delta \hat{\bf{g}}_s^e {\bf{V}}_k (\Delta \hat{\bf{g}}_s^e)^H+2{\rm{Re}}\{\hat{\bf{g}}_s^e {\bf{V}}_k (\Delta \hat{\bf{g}}_s^e)^H\}+\hat{\bf{g}}_s^e {\bf{V}}_k (\hat{\bf{g}}_s^e)^H-\psi_k\leq  0,\label{QQ1}\\
&\Delta \hat{\bf{g}}_s^e {\bf{\Xi}}_i (\Delta \hat{\bf{g}}_s^e)^H\!+\!2{\rm{Re}}\{\hat{\bf{g}}_s^e {\bf{\Xi}}_i (\Delta \hat{\bf{g}}_s^e)^H\}\!+\!\hat{\bf{g}}_s^e {\bf{\Xi}}_i (\hat{\bf{g}}_s^e)^H\!+\!\phi_k\!-\!\sigma_s^2\leq 0,\label{QQ2}\\
&\Delta \hat{\bf{g}}_s^e(\Delta \hat{\bf{g}}_s^e)^H-\tau_s\leq 0,\label{QQ3}\\
&\psi_k\leq \kappa_k^2, \kappa_k^2\leq \beta_k\phi_k,\label{QQ4}
\end{align}
\end{subequations}
where ${\bf{\Xi}}_i=-\sum_{i\neq k}^K{\bf{V}}_i$. Combing $Lemma$ 1, (\ref{QQ1}) and (\ref{QQ3}), we can obtain the following convex linear matrix inequality (LMI)
\begin{eqnarray}\label{LMI1}
\left[ \begin{array}{ccc}
\gamma_k{\bf{I}}-{\bf{V}}_k & -(\hat{\bf{g}}_s^e {\bf{V}}_k)^H \\
-\hat{\bf{g}}_s^e {\bf{V}}_k & \psi_k-\gamma_k\tau_s-\hat{\bf{g}}_s^e {\bf{V}}_k (\hat{\bf{g}}_s^e)^H
\end{array} \right ]\succeq {\bf{0}}.
\end{eqnarray}  
Similarly, (\ref{QQ2}) can be recast into the following convex LMI:
\begin{eqnarray}\label{LMI2}
\left[ \begin{array}{ccc}
\varepsilon_k{\bf{I}}-{\bf{\Xi}}_k & -(\hat{\bf{g}}_s^e {\bf{\Xi}}_k)^H \\
-\hat{\bf{g}}_s^e {\bf{\Xi}}_k & \sigma_s^2-\phi_k-\varepsilon_k\tau_s-\hat{\bf{g}}_s^e {\bf{\Xi}}_k (\hat{\bf{g}}_s^e)^H
\end{array} \right ]\succeq {\bf{0}}.
\end{eqnarray}  
In addition, $\kappa_k^2\leq \beta_k\phi_k$ can be expressed in the following matrix form according to the Schur complement lemma~\cite{Zhang_TVT_2017} 
\begin{eqnarray}\label{LMI3}
\left[ \begin{array}{ccc}
\beta_k& \kappa_k \\
\kappa_k & \phi_k
\end{array} \right ]\succeq {\bf{0}}.
\end{eqnarray}  
According to the first-order Taylor approximation formula, $\kappa_k^2\approx (\kappa_k^{[i]})^2+2(\kappa_k-\kappa_k^{[i]})\kappa_k^{[i]}$. Therefore,  $\psi_k\leq \kappa_k^2$ in (\ref{QQ4}) can be written as the following convex constraint:
\begin{eqnarray}\label{LMI4}
   (\kappa_k^{[i]})^2+2(\kappa_k-\kappa_k^{[i]})\kappa_k^{[i]}\geq \psi_k, k\in{\mathcal{K}}.
\end{eqnarray}

Finally, we formulate the SDP problem as
\begin{subequations}\label{OptG}
	\begin{align}
	\;\;\;\;\;\;\;\;\;\;\;\;\;\;\;\;\;\;\;\;&\underset{\left\{{\bf{V}}_k,\alpha_{k},\beta_k,\mu_k,\gamma_k,\varepsilon_k,\kappa_k,\phi_k,\psi_k\right\}}{\rm{min}}\;\;t\;\;\label{OptG0}\\
	{\rm{s.t.}}\;\;&\sum\nolimits_{k=1}^{K}{\bf{V}}_k(l,l)\leq P_{\rm{max}}^l, l\in \mathcal{L},\label{OptG1}\\
	&{\rm{Rank}}({\bf{V}}_k)=1,k\in \mathcal{K},\label{OptG2}\\
	&{\rm{(\ref{E1}),(\ref{Eq2}),(\ref{Eq3}),(\ref{LMI1}),(\ref{LMI2}),(\ref{LMI3}),({\ref{LMI4}})}}.
	\end{align}
\end{subequations} 

Obviously, (\ref{OptG}) is a non-convex SDP due to the rank-one constraint. However, by removing (\ref{OptG2}), the above problem becomes a convex SDP and can be solved efficiently by numerical solvers such as SDPT3~\cite{toh1999sdpt3}.  To obtain the solution of ${\mathbb{P}}2$, we need to iteratively solve the relaxed problem of  (\ref{OptG}). The procedure  is summarized as Algorithm~\ref{algorithm2}.
\begin{algorithm}[t]
	{\caption{The Proposed Iterative Algorithm for Solving~${\mathbb{P}}2$.}
		\label{algorithm2}
		{\bf{Initialize}} $\alpha_k^{[i]}$, $\beta_k^{[i]}$, $\mu_k^{[i]}$, $\kappa_k^{[i]}$, $i=1$, the maximum iteration index $I_{\rm{max}}$.\\
		\Repeat{$i=T_{\rm{max}}$ {\rm{or} Convergence}}{
			Solve the relaxed problem (\ref{OptG}) and obtain the solution ${\bf{V}}_k^{\star},\alpha_{k}^{\star},\beta_k^{\star},\mu_k^{\star},\gamma_k^{\star},\varepsilon_k^{\star},\kappa_k^{\star},\phi_k^{\star},\psi_k^{\star}$.\\
		Update $i\leftarrow i+1$.\\
		Update $\alpha_k^{[i]}\leftarrow \alpha_k^{\star}$, $\beta_k^{[i]}\leftarrow \beta_k^{\star}$, $\mu_k^{[i]}\leftarrow \mu_k^{\star}$, $\kappa_k^{[i]}\leftarrow \kappa_k^{\star}$.
				}
}
\end{algorithm} 
Meanwhile, we have the following proposition.
\begin{proposition}\label{Pro1}
	The objective function of (\ref{OptG}) is a non-increasing sequence at each iteration based on the proposed Algorithm~\ref{algorithm2}, and it  converges to a stationary solution.
\end{proposition}
\begin{proof}
	To prove proposition~\ref{Pro1},  we first need to confirm that the solution of problem~(\ref{OptG}) at the $i$th iteration is also a feasible solution for the $(i+1)$th iteration. We assume that ${\bf{V}}_k^{\star},\alpha_{k}^{\star},\beta_k^{\star},\mu_k^{\star},\gamma_k^{\star},\varepsilon_k^{\star},\kappa_k^{\star},\phi_k^{\star},\psi_k^{\star}$ are the optimal solutions of problem~(\ref{OptG}) at the $i$th iteration.  To proceed with $\mathbb{P}$2, the convex approximated techniques are adopted for constraints (\ref{E1}), (\ref{Eq3}) and (\ref{LMI4}). Therefore, we need to prove that  those constraints still hold at the $(i+1)$th iteration  for the solutions obtained from the $i$th iteration. To facilitate analysis, we define the following function
	\begin{eqnarray}
		f(\beta_k^{[i]})=\log(1+\beta_{k}^{[i]})+\frac{\beta_{k}^\star-\beta_{k}^{[i]}}{1+\beta_{k}^{[i]}}.
	\end{eqnarray}
	
By replacing the variable $\beta_k^{[i]}$ at the $(i\!+\!1)$th iteration with the value obtained at the $i$th iteration,  namely   $\beta_k^{[i\!+\!1]}\!=\! \beta_k^{\star}$, we have
	\begin{subequations}
	\begin{align}
\;\;\;\;\;\;\;f(\beta_k^{[i+1]})&=\log(1+\beta_{k}^{[i+1]})+\frac{\beta_{k}^\star-\beta_{k}^{[i+1]}}{1+\beta_{k}^{\star}}\\
&=\log(1+\beta_{k}^\star)\\
&\geq \log(1+\beta_{k}^{[i]})+\frac{\beta_{k}^\star-\beta_{k}^{[i]}}{1+\beta_{k}^{[i]}}.\label{QQQ1}
\end{align}
\end{subequations}
where (\ref{QQQ1}) is obtained by the first-order Taylor approximation. Therefore, we show that (\ref{E1}) still holds at the $(i\!+\!1)$th iteration  for the solutions obtained from the $i$th iteration. The same conclusions can be obtained for (\ref{Eq3}) and (\ref{LMI4}) following a similar procedure, and they are thus omitted here. As a result, the solution of problem~(\ref{OptG}) at the $i$th iteration is also a feasible solution for the $(i\!+\!1)$th iteration. 

Since problem~(\ref{OptG}) is convex, the objective function value achieved at each iteration will decrease or at least maintain the value achieved at the previous iteration.  Due to  the limited transmit power, the objective function value has a lower bound and converges to a stationary solution. 
\end{proof}

\begin{table}[t]
	\caption{{Ratio of Rank-one Solutions.}}
	\label{table1}
	\begin{center}
		\begin{tabular}{|c|c|c|c|}
			\hline
			$P_{\rm{max}}^l$ & 20 dBm & 30 dBm& 40 dBm \\\hline
			Rank-one&991 & 988 & 993 \\\hline
			Ratio&99.1\%&98.8\% &99.3\%\\\hline
		\end{tabular}
	\end{center}
\end{table} 

Finally, we need to consider whether the obtained solution satisfies the rank-one constraint or not. The  rank-one solution is satisfied if the following condition holds~\cite{Ni_TCOM_2019}:
\begin{eqnarray}\label{sdp}
	\frac{\Upsilon_{\rm{max}}({\bf{V}}_k)}{{\rm{Tr}}({\bf{V}}_k)}=1,k\in\mathcal{K},
\end{eqnarray}
where $\Upsilon_{\rm{max}}({\bf{V}}_k)$ denotes the maximum eigenvalue of ${\bf{V}}_k$. To verify (\ref{sdp}), 
we perform 1,000 times simulations, and the solutions are summarized in Table~\ref{table1}.  From Table~\ref{table1}, one can observe that the probability of a rank-one solutions reaches up to almost 99\%. Therefore, we can most likely obtain the rank-one solutions via the proposed algorithm. Even when the solutions are not rank-one,   several advanced methods can be applied  to reconstruct rank-one solutions, e.g., hybrid beamforming design scheme~\cite{Ng_TVT_2016} and randomization beamforming design scheme~\cite{Luo_TSP_2006}.


\section{Numerical Results}
In this section, numerical results are presented to evaluate the performance of the proposed schemes. For simplicity, we assume that all RRHs have the same maximum transmit power, i.e., $P_{\rm{max}}^l\!=\!P_{\rm{max}}$. All noise powers are assumed to be the same, i.e., $\sigma_k^2\!=\!\sigma_s^2\!=\!\sigma^2$. All RRHs, actuators and Eves are uniformly distributed within a circular cell with 100~m radius. The distance between the CP and the cell center is 300~m.  We assume that the channel between the CP and RRHs has a single path, and the path loss is modeled as $1/(1+(d_l/d_0)^\varrho)$~\cite{He_TWC_2019}, where $d_l$, $d_0$ and $\varrho$ denote the distance between the CP and RRH~$l$, reference distance and the pathloss exponent, respectively. The AoD/AoA are assumed uniformly distributed within $[0,\;2\pi]$. In addition, the channels from the RRHs to actuators and Eves contain $G=4$ paths, and the pathloss is modeled as $1/(1+(d_{l,k}/d_0)^\varsigma)$. The mmWave bandwidth is set to be 1 GHz. The default simulation parameters are listed in Table~\ref{table2}.  Unless otherwise specified, these default values are used in simulation.
\begin{table}[t]
	\caption{{Simulation Parameters}}
	\label{table2}
	\begin{center}
		\begin{tabular}{|c|c||c|c||c|c|}
			\hline
			Symbol & Value & Symbol& Value&Symbol&Value \\\hline
			$M$&120 & $N$ & 6 & $L$ & 6 \\\hline
			$K$&4   & $S$ & 2 & $\sigma^2$& 0.01\\\hline
			$d_0$& 10 m &$\varrho$ & 2 &$\varsigma$ & 4 \\\hline
			$G$ & 4 & $P$ & 46 dBm &$P_{\rm{max}}$& 40 dBm \\\hline 
			$U$&10&-&-&-&-\\\hline 
		\end{tabular}
	\end{center}
\end{table} 

\begin{figure}[t]
	\begin{center}
		\includegraphics[width=8cm,height=6cm]{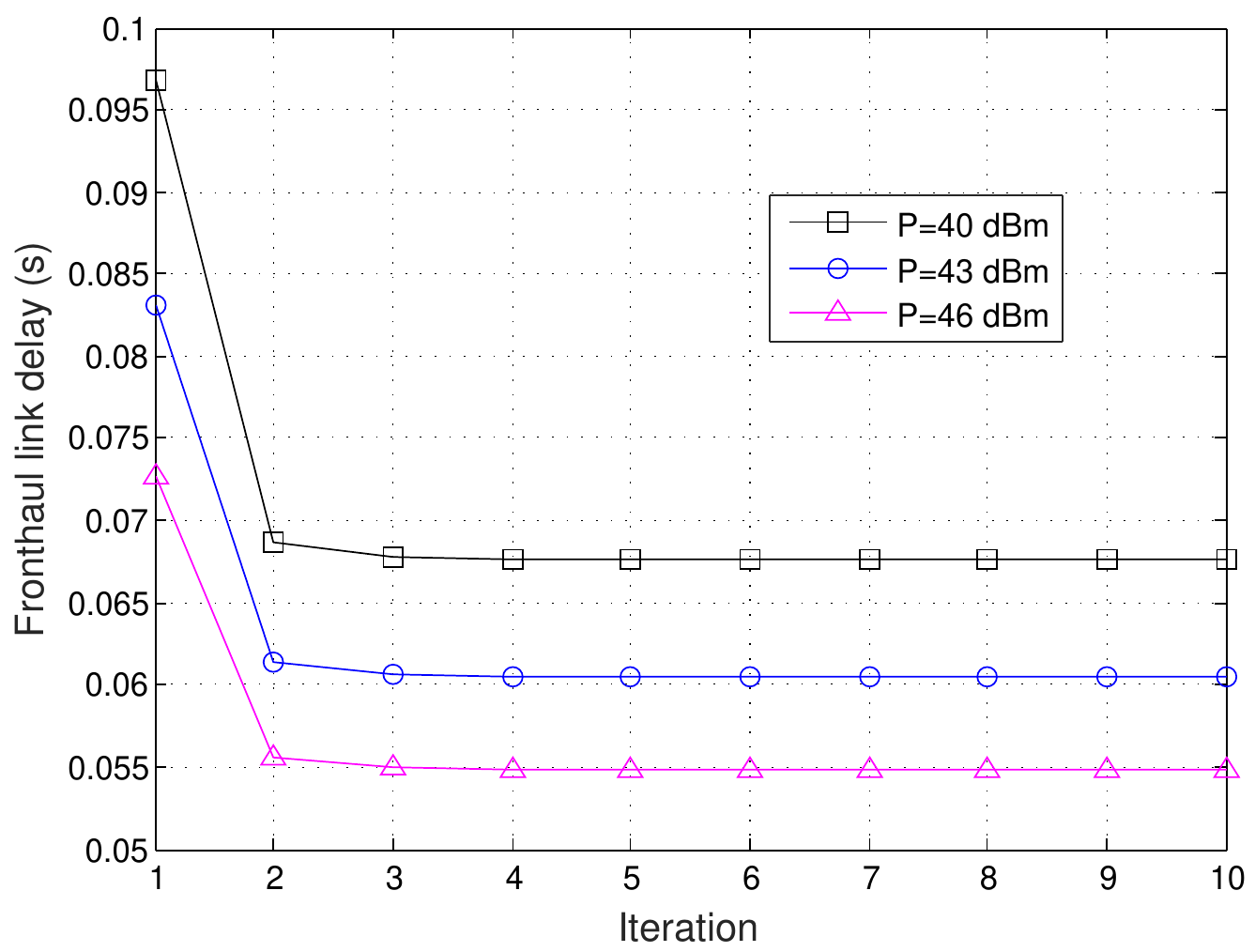}
		\caption{The fronthaul link delay  versus iteration.}
		\label{Figuer1}
	\end{center}
\end{figure}

\begin{figure}[t]
	\begin{center}
		\includegraphics[width=8cm,height=6cm]{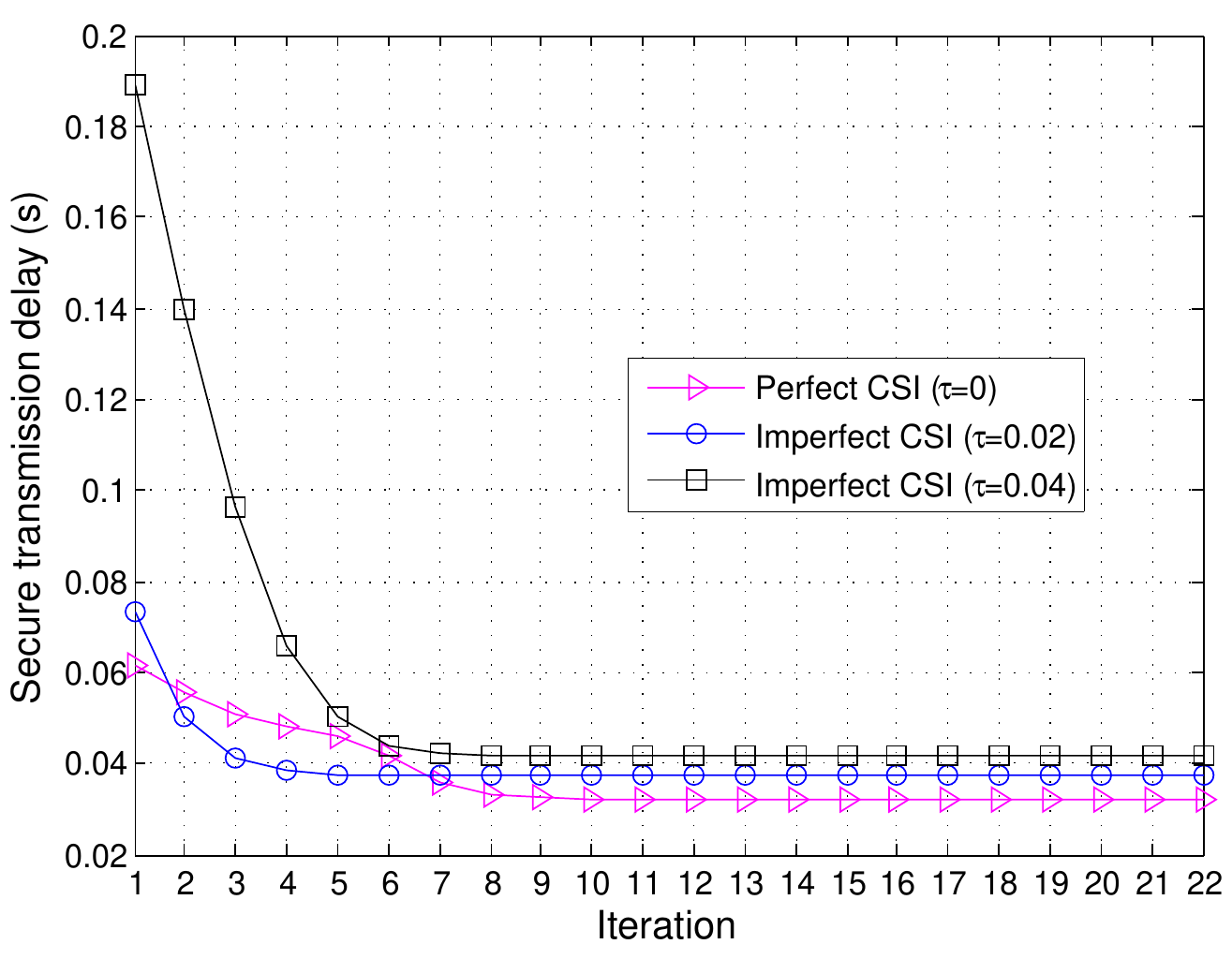}
		\caption{The secure transmission delay versus iteration.}
		\label{Figuer2}
	\end{center}
\end{figure}

\begin{figure}[t]
	\begin{center}
		\includegraphics[width=8cm,height=6cm]{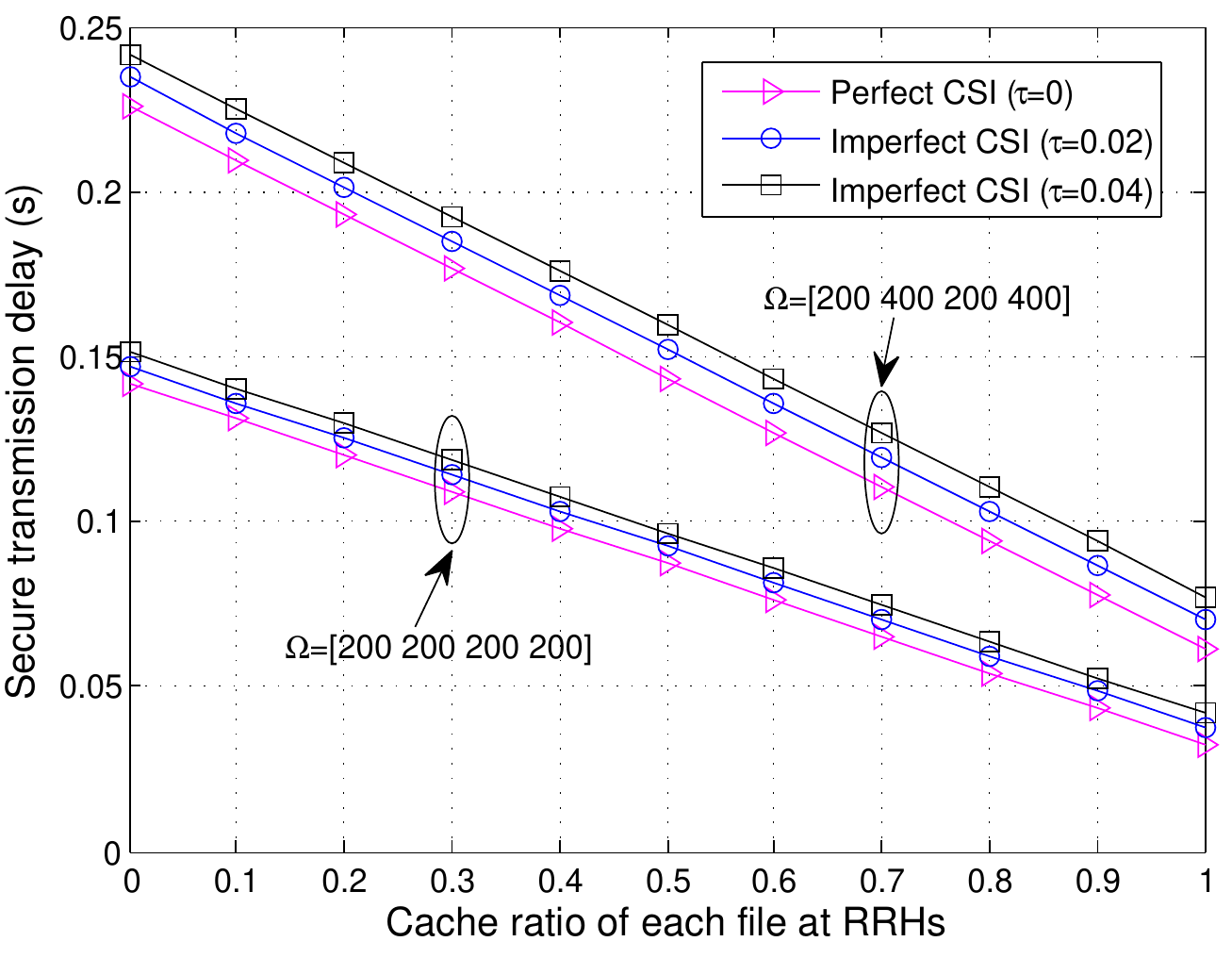}
		\caption{The secure transmission delay versus cache ratio of each file at the RRHs.}
		\label{Figuer3}
	\end{center}
\end{figure}
\begin{figure}[t]
	\begin{center}
		\includegraphics[width=8cm,height=6cm]{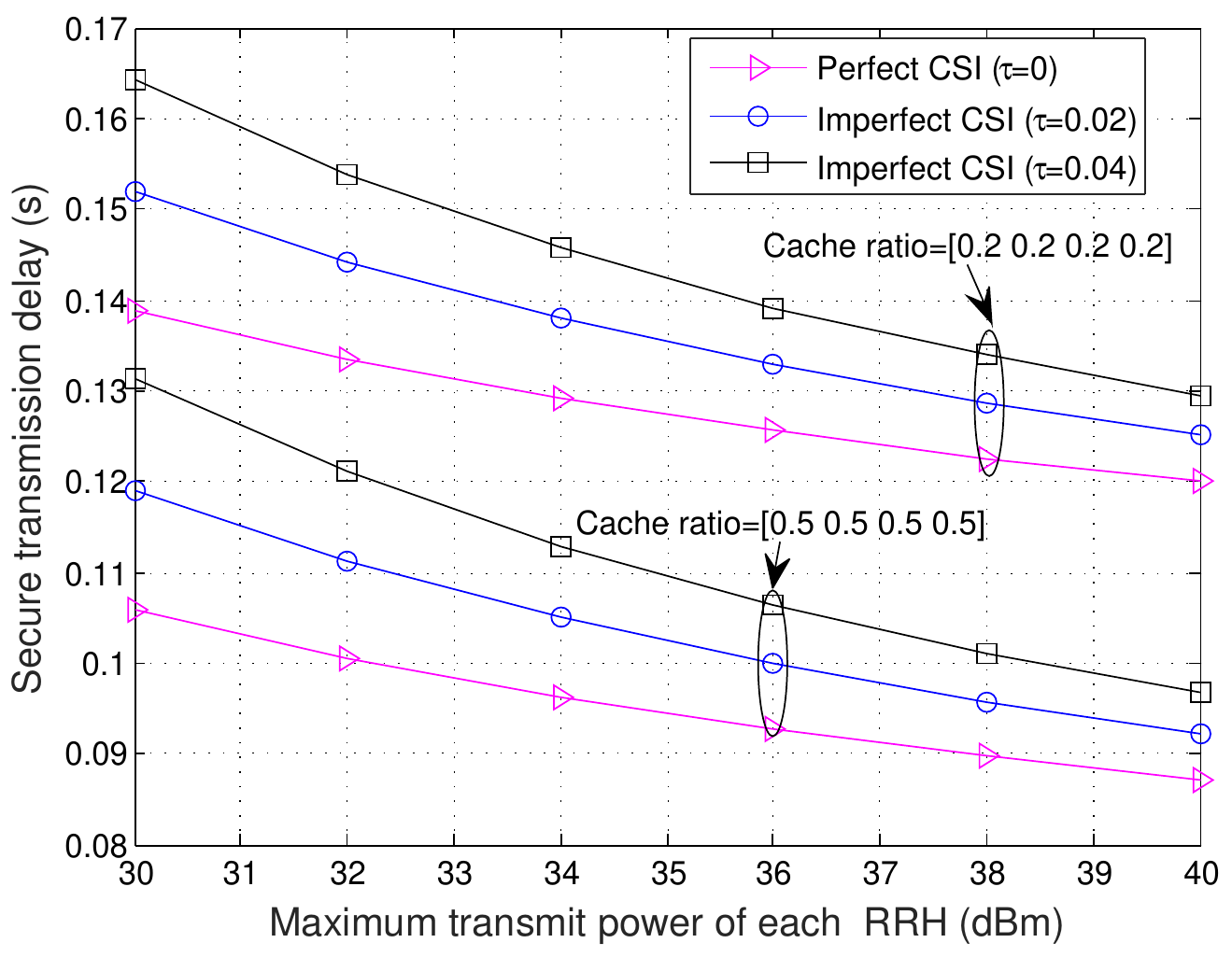}
		\caption{The secure transmission delay versus maximum transmit power of each RRH with $\Omega$=[200 200 200 200].}
		\label{Figuer4}
	\end{center}
\end{figure}
\begin{figure}[t]
	\begin{center}
		\includegraphics[width=8cm,height=6cm]{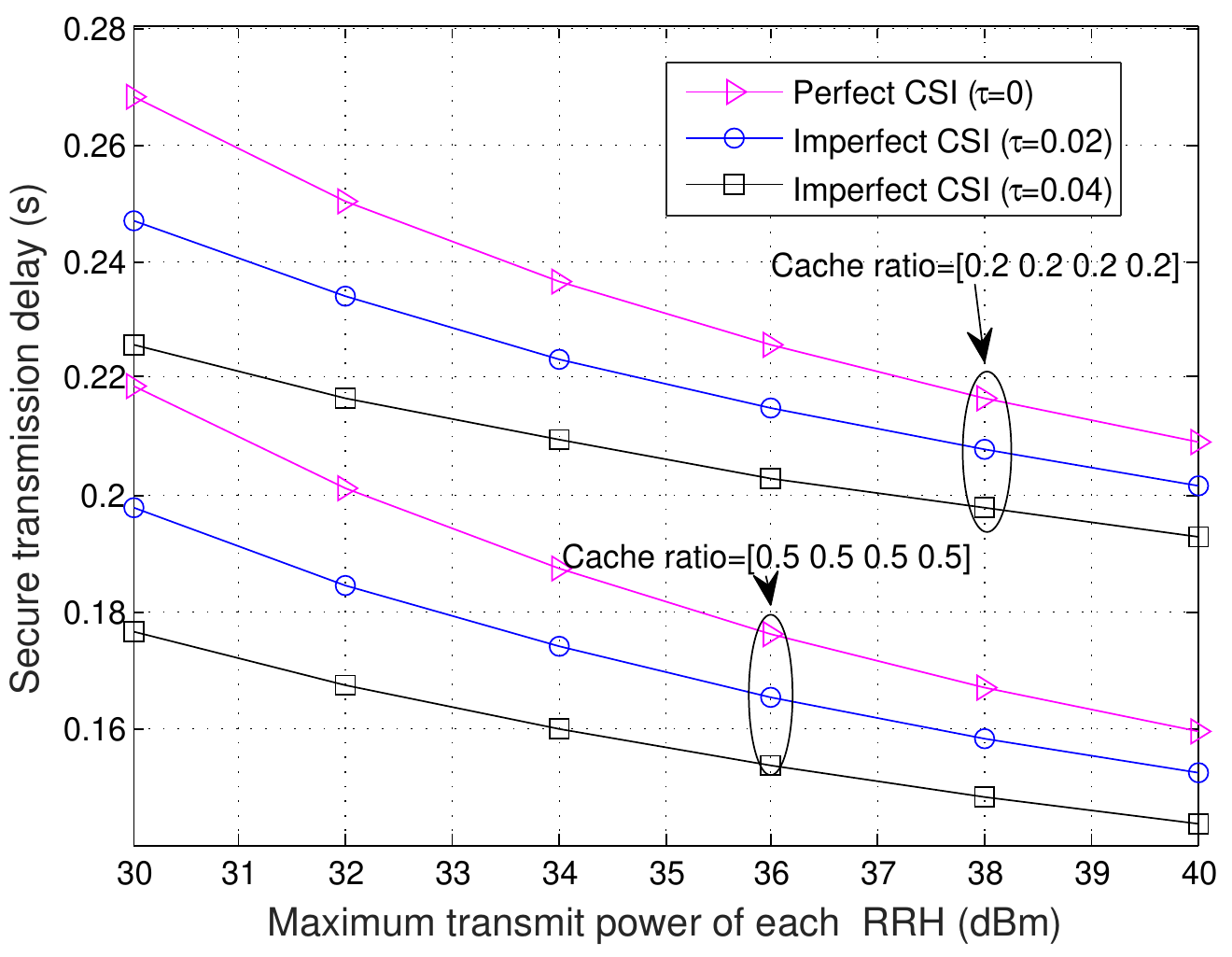}
		\caption{The secure transmission delay versus maximum transmit power of each RRH with $\Omega$=[200 400 200 400].}
		\label{Figuer5}
	\end{center}
\end{figure}
\begin{figure}[t]
	\begin{center}
		\includegraphics[width=8cm,height=6cm]{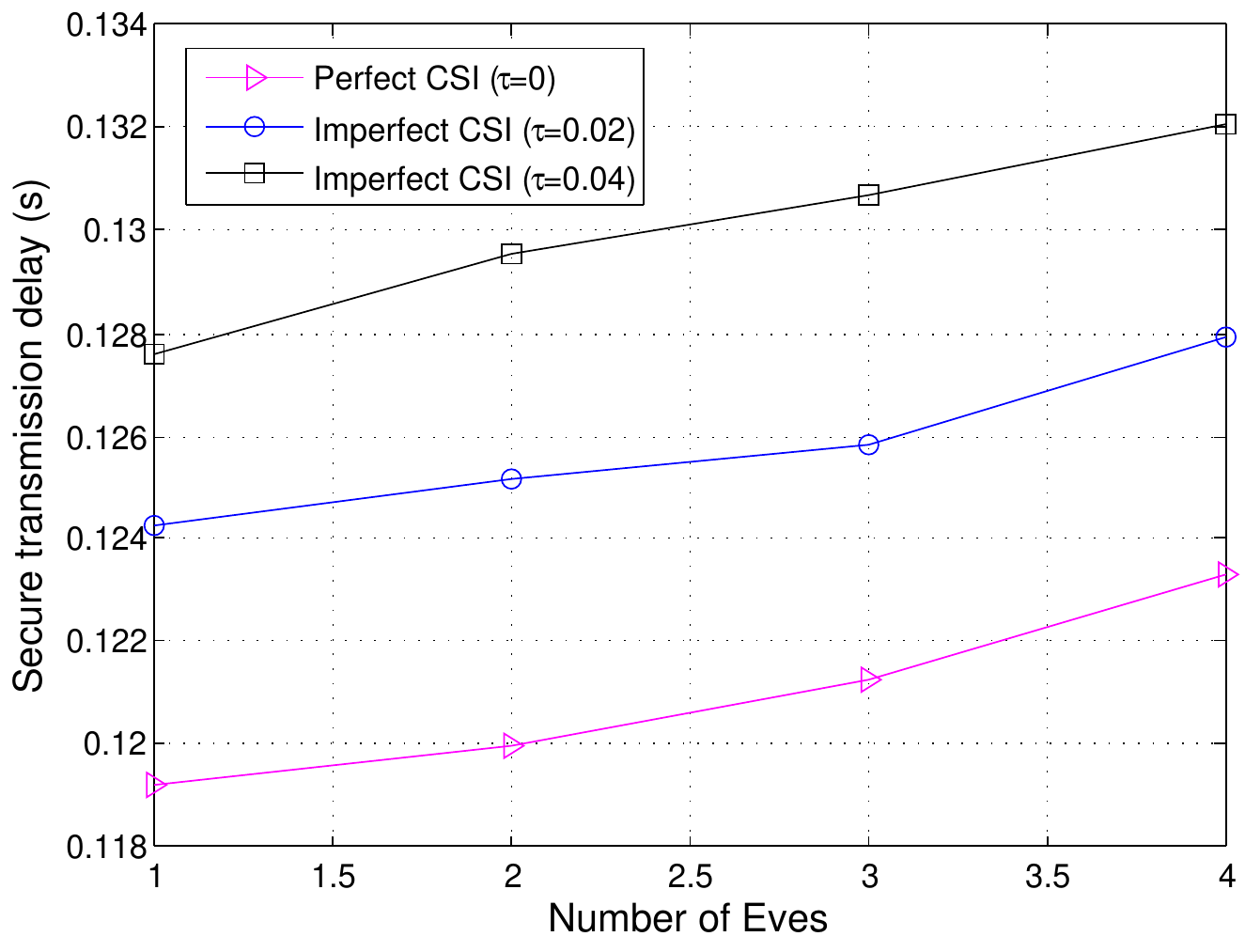}
		\caption{The secure transmission delay versus number of eavesdroppers with cache ratio=[0.2 0.2 0.2 0.2].}
		\label{Figuer6}
	\end{center}
\end{figure}
\begin{figure}[t]
	\begin{center}
		\includegraphics[width=8cm,height=6cm]{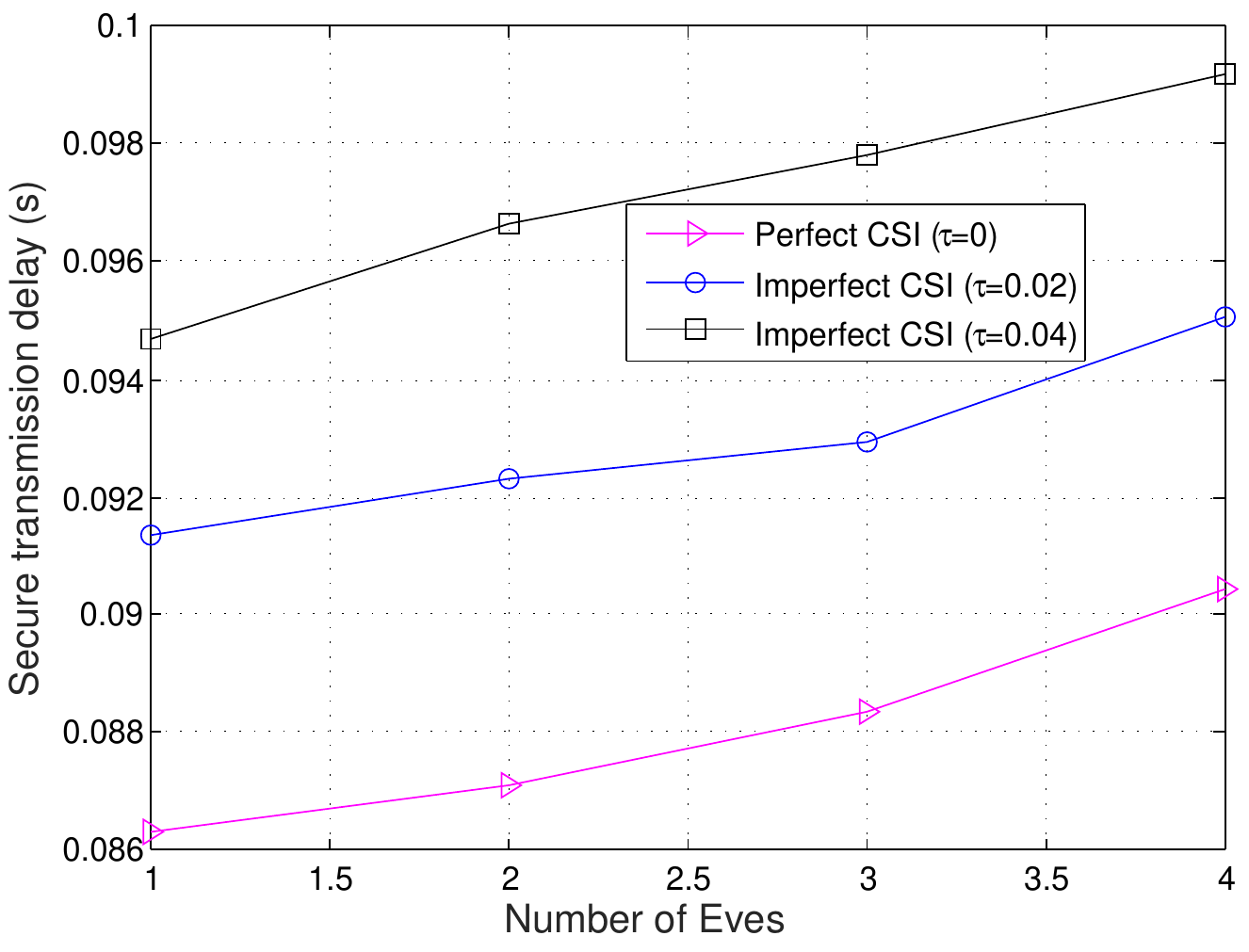}
		\caption{The secure transmission delay versus number of eavesdroppers with cache ratio=[0.5 0.5 0.5 0.5].}
		\label{Figuer7}
	\end{center}
\end{figure}

Figure~\ref{Figuer1} shows the convergence property of the proposed alternatively iterative algorithm for solving ${\mathbb{P}}$1 under different CP transmit powers. We set the file size requested by each actuator to $200$ Mbits, and half of each file is cached in the RRHs. One can observe that it takes only about 4 iterations for the proposed algorithm to converge.  We also plot the convergence speed of the proposed Algorithm~\ref{algorithm2} for solving ${\mathbb{P}}$2 in Fig.~\ref{Figuer2}.  We consider both the perfect CSI scenario when $\tau=0$, and two imperfect ones, with different error variance.  Note that Fig.~\ref{Figuer2} only plots the secure transmission delay from the RRHs to the actuators. It is clear that the secure transmission delay first decreases and then converges to a certain value in all scenarios. As expected, the secure transmission delay decreases with $\tau$, and achieves the lowest value under the perfect CSI scenario. 

Figure~\ref{Figuer3} plots the secure transmission delay versus cache ratio of each file at the RRHs. Here, we consider two different file requirements for the actuators: $\Omega=$[200 200 200 200] denotes the case that the file sizes requested by each actuator are $200$ Mbits, while  $\Omega=$[200 400 200 400] represents that when the file sizes requested by the four actuators are $200$ Mbits, $400$ Mbits, $200$ Mbits, $400$ Mbits, respectively. In~Fig.~\ref{Figuer3}, ``0" at the horizontal axis means that the requested files by the actuators are not cached in the RRHs at all, while ``1" means that the files requested by the actuators are all cached in the RRHs. One can see that the secure transmission delay decreases linearly with the cache ratio for all schemes. The reason is that a higher cache ratio reduces the data transmission  from the CP to RRHs, and consequently, lowers the fronthaul links delay.    

In~Fig.~\ref{Figuer4}, we show the secure transmission delay versus the maximum transmit power of each RRH under different cache ratios, where  $\Omega=$[200 200 200 200]. ``Cache ratio=[0.2 0.2 0.2 0.2]" means that 20\% of requested file by each actuator is cached in  the RRHs, while ``Cache ratio=[0.5 0.5 0.5 0.5]" means that 50\% of requested file by each actuator is cached in  the RRHs. One can see that the secure transmission delay decreases with the maximum transmit power of each RRH. In addition, a higher cache ratio leads to a smaller delay, which is also evidenced in~Fig.~\ref{Figuer3}. Fig.~\ref{Figuer5} shows the results when  $\Omega=$[200 400 200 400]. We can reach the same conclusions as in Fig.~\ref{Figuer4}, and the only difference is that the secure transmission delay is higher due to the larger files requested by actuators. 

We show the secure transmission delay versus the number of Eves in Fig.~\ref{Figuer6}. Here, we set cache ratio=[0.2 0.2 0.2 0.2] and the required file size  $\Omega=$[200 200 200 200]. It can be seen that the secure transmission delay increases with the number of Eves. This coincides with~(\ref{securerate}), which indicates that the presence of more Eves results in a lower secrecy transmission rate.  Additionally, a large CSI estimation error also leads to higher secure transmission delay. We also show results under cache ratio=[0.5 0.5 0.5 0.5] in Fig.~\ref{Figuer7}. One can observe that the secure transmission delay is much lower since more files requested can be directly fetched from the RRHs.

Note that for the specific values of simulation parameters, e.g., estimated channel error bound, each RRH's transmit power, etc., they are adopted as references, and may not be the same for the real network. Nonetheless, these specific values do not affect the evaluation of the proposed algorithm, and the trend still holds for other values. Once the parameter values are given, we can directly obtain the optimal beamforming and system performance based on the proposed algorithm.  For example, a larger channel estimation error leads to a higher transmission delay, and a larger RRH's transmit power brings a small transmission delay.
\section{Conclusion}
In this paper, we investigated a two-phase secure transmission delay minimization problem in an edge cache-assisted mmWave C-RAN. At the first transmission phase, a joint transmit beamforming at the CP and receive beamfroming at the RRHs scheme was proposed. At the second transmission phase, we first designed the analog beamforming for each RRH, and then transformed the formulated problem into a series of convex subproblems by SCA technique, $S$-procedure and SDP relaxation. Finally, an iterative algorithm was proposed, which converges to at least a local optimum. The presented simulation results show that the solutions have the rank-one characteristic with a high probability (near 99\%).   Meanwhile, the proposed algorithms have been shown to achieve fast convergence. Moreover, a detailed illustration has been provided to demonstrate how the secure transmission delay is affected by the different system parameters, e.g., the maximum transmit power at the RRHs, the cache-ratio, the channel estimation error, and the number of Eves. These results can be used as references during system design, where different tradeoffs need to be considered.
\bibliographystyle{ieeetr}
\bibliography{references}
\end{document}